\documentclass[11pt,a4paper]{article}
\usepackage{fullpage}
\usepackage{authblk}
\usepackage{listings}
\usepackage{enumerate}
\lstnewenvironment{datalogpgm}[1][]
  {\lstset{#1}}
  {}

\usepackage{amsmath, amssymb, amsthm, verbatim}

\usepackage{complexity}
\usepackage{tikz}
\usetikzlibrary{automata,positioning}
\usepackage{pgfplots}
\pgfplotsset{compat=1.10}
\usetikzlibrary{shapes.geometric,arrows,fit,matrix,positioning}
\tikzset
{
    treenode/.style = {ellipse, draw=black, align=center, minimum size=1cm}
}

\newcommand{\IN}{\mathbb{N}}
\newcommand{\INpos}{\IN_{0}}

\newcommand{\card}[1]{|{#1}|}

\mathchardef\mhyphen="2D
\newcommand{\hph}[1]{{\mathsf{H\mhyphen PARTITION}}({#1})}
\newcommand{\hp}{{\mathsf{H\mhyphen PARTITION}}}
\newcommand{\sat}{{\mathsf{SAT}}}
\newcommand{\twosat}{{\mathsf{2SAT}}}

\newcommand{\mg}[1]{({#1})}
\newcommand{\Ndot}[1]{\mathsf{N_{D}}({#1})}
\newcommand{\Nfull}[1]{\mathsf{N_{F}}({#1})}
\newcommand{\Full}[1]{{\mathsf{Full}}({#1})}
\newcommand{\Dotted}[1]{{\mathsf{Dotted}}({#1})}

\newcommand{\model}[1]{({#1})}

\newtheorem{thm}{Theorem}

\newtheorem{prop}[thm]{Proposition}
\newtheorem{lem}[thm]{Lemma}
\newtheorem{corol}[thm]{Corollary}
\theoremstyle{definition}

\newtheorem{definition}{Definition}
\newcommand{\maxi}{\mathit{max}}
\newcommand{\mini}{\mathit{min}}

\title{Computing H-Partitions in ASP and Datalog}
\author{Chloé Capon\\
\texttt{chloe.capon2@student.umons.ac.be}
\and 
Nicolas Lecomte\\
\texttt{nicolas.lecomte@student.umons.ac.be}
\and Jef Wijsen\\
\texttt{jef.wijsen@umons.ac.be}}
\affil{University of Mons (UMONS)\\
Mons, Belgium}
\date{}

\begin{document}

\maketitle


\begin{abstract}
A $H$-partition of a finite undirected simple graph $G$ is a labeling of~$G$'s vertices such that the constraints expressed by the model graph~$H$ are satisfied. For every model graph~$H$, it can be decided in non-deterministic polynomial time whether a given input graph~$G$ admits a $H$-partition. Moreover, it has been shown in~\cite{DBLP:journals/ita/DantasFGK05} that for most model graphs, this decision problem is in deterministic polynomial time. 
In this paper, we show that these polynomial-time algorithms for finding $H$-partitions can be expressed in Datalog with stratified negation.
Moreover, using the answer set solver Clingo~\cite{DBLP:conf/lpnmr/GebserKKS11, DBLP:journals/corr/GebserKKS14}, we have conducted experiments to compare straightforward guess-and-check programs with Datalog programs. 
Our experiments indicate that in Clingo, guess-and-check programs run faster than their equivalent Datalog programs.
\end{abstract}

\maketitle


\section{Introduction}

Answer Set Programming (ASP) is a powerful programming paradigm that allows for an easy encoding of decision problems in $\NP$. 
If the answer to a problem in $\NP$ is ``yes,''  then, by definition, there is a ``yes''-certificate that can be checked in polynomial time. 
In an ASP \emph{guess-and-check program}, a programmer first declares the format of such a certificate, and then specifies the constraints that a well-formatted certificate should obey in order to be a ``yes''-certificate. 
For example, for the well-known problem $\sat$, an ASP-programmer can first declare that certificates take the form of truth assignments, and then specify that ``yes''-certificates are those certificates that leave no clause unsatisfied. 

While ASP guess-and-check programs are typically oriented towards $\NP$-complete problems, they can also be used for problems in $\P$. For example, the previously mentioned encoding of $\sat$ also solves $\twosat$, which is known to be in $\P$.
This raises the following issue which will be addressed in this paper.
Assume that we have an answer set solver at our disposal, and that we have written a guess-and-check ASP program for a particular problem that is $\NP$-complete in general (for example, $\sat$).
Assume furthermore that we know that under some restrictions, the problem can be solved in polynomial time (for example, the restriction of $\sat$ to $\twosat$). In logic programming, it may be possible to encode such a polynomial-time solution in a syntactic restriction of ASP, for example, in Datalog with stratified negation~\cite{DBLP:books/aw/AbiteboulHV95}, which guarantees the existence of a polynomial-time execution. 
From a theoretical standpoint, a program in Datalog is more efficient than a general guess-and-check ASP program.
From a practical standpoint, however, it is not clear whether such a Datalog program will run faster in our ASP engine, compared to a general guess-and-check ASP program. 

%

Our problem of interest is that of finding $H$-partitions in undirected graphs~$G$~\cite{DBLP:journals/ita/DantasFGK05}. 
All undirected graphs in this paper are understood to be finite and simple.
Given an undirected graph $G$, this problem consists in deciding whether there exists a partition of $G$'s vertices such that this partition respects some constraints encoded in another graph $H$ with four vertices, called \emph{model graph}. Each class of the partition is labeled by a vertex in $H$. There are two types of constraints:
\begin{itemize}
\item 
if $H$ contains a \emph{full edge} between two vertices, $A$ and $B$, then each vertex in the class labeled by $A$ must be adjacent to each vertex in the class labeled by $B$; and
\item 
if $H$ contains a \emph{dotted edge} between two vertices, $A$ and $B$, then each vertex in the class labeled by $A$ must be nonadjacent to each vertex in the class labeled by $B$.
\end{itemize}
Note that there are no constraints between two vertices in the same class. 
All possible model graphs, up to isomorphisms, are presented in Figures~\ref{fig:k2s22k2} and~\ref{fig:list_models}.
The $H$-partitioning problem for $K_{2}+S_{2}$, also known as finding a skew partition, is in polynomial time~\cite{DBLP:journals/jal/FigueiredoKKR00,DBLP:conf/jcdcg/KennedyR07}.
The $H$-partitioning problem for $2K_{2}$ is $\NP$-complete~\cite{DBLP:journals/endm/CamposDFG05}.
For each model graph in Figure~\ref{fig:list_models}, Dantas et al.~\cite{DBLP:journals/ita/DantasFGK05} provide a polynomial-time algorithm, of low polynomial degree, for deciding whether an undirected graph admits a $H$-partition.
In this paper, we show how these polynomial-time algorithms can be encoded in  Datalog with stratified negation.
We then experimentally compare these Datalog programs with a guess-and-check ASP program.
In theory, the computational complexity of Datalog is in $\P$, while guess-and-check ASP programs can solve $\NP$-complete problems.

\begin{center}
\begin{figure}\centering
\includegraphics[scale=0.7]{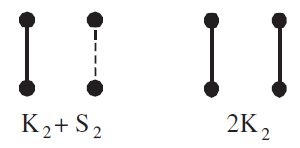}
\caption{Examples of model graphs (figure copied from \cite{DBLP:journals/ita/DantasFGK05}).}
\label{fig:k2s22k2}
\end{figure}
\end{center}

This paper is organized as follows. 
Section~\ref{sec:prelim} formalizes the problem of finding $H$-partitions and paraphrases the polynomial-time method of Dantas et al.~\cite{DBLP:journals/ita/DantasFGK05}. 
Sections~\ref{sec:asp} and~\ref{sec:datalog_strat} focus on finding $H$-partitions in logic programming, first in ASP, and then in Datalog with stratified negation.
In our experimental validation of Section~\ref{sec:gen_instances}, we first generate undirected graphs of different sizes on which our programs can be executed.
We distinguish between those input graphs that admit a $H$-partition, called yes-instances, and those that do not, called no-instances.
We show different methods for generating yes-instances and no-instances, for a fixed model graph $H$, and then compare the running times of our logic programs in the answer set solver Clingo~\cite{DBLP:conf/lpnmr/GebserKKS11, DBLP:journals/corr/GebserKKS14}.  
Finally, Section \ref{sec:conclusion} concludes the paper.

\section{Preliminaries} \label{sec:prelim}

In this section, we formally define the problem $\hp$ and sketch an algorithm borrowed from~\cite{DBLP:journals/ita/DantasFGK05}. 
In Section~\ref{sec:datalog_strat}, this algorithm will be written in Datalog with stratified negation.

\subsection{The Problem $\hp$}

\begin{figure}\centering
\includegraphics[scale=0.7]{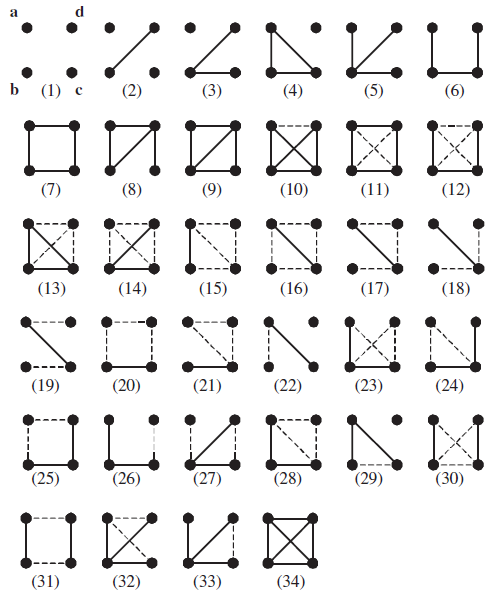}
\caption{List of model graphs (figure copied from \cite{DBLP:journals/ita/DantasFGK05}).}
\label{fig:list_models}
\end{figure}

A \emph{model graph} $H$ is an undirected graph with four vertices, called $A$, $B$, $C$ and $D$ (sometimes written in lowercase letters). Every edge is of exactly one of two types: \emph{full} or \emph{dotted}.
All possible model graphs, up to isomorphisms, are presented in Figures~\ref{fig:k2s22k2} and~\ref{fig:list_models}.
We define $\hp$ as the following decision problem.
\begin{description}
\item[Problem $\hp$]
\item[Input:]
A model graph $H$; an undirected graph $G$.
\item[Question:]
Is it possible to partition $V(G)$ into four pairwise disjoint, non-empty subsets, $V_A$, $V_B$, $V_C$, and $V_D$ such that, for all $p,q\in\{A, B, C, D\}$ with $p \neq q$:
\begin{itemize}
\item 
if $(p,q)$ is a full edge of $H$, then every vertex in $V_p$ is adjacent to every vertex in $V_q$; and
\item 
if $(p,q)$ is a dotted edge of $H$, then every vertex in $V_p$ is nonadjacent to every vertex in $V_q$.
\end{itemize}
Such a partition, if it exists, is called a \emph{solution} or a \emph{$H$-partition}. 
\end{description}
We call a pair $(H,G)$ a \emph{yes-instance} if $\hp$ returns \emph{yes} on input~$H$ and~$G$; otherwise $(H,G)$ is a \emph{no-instance}.
$\hph{H}$ denotes the $\hp$ problem for a fixed model graph~$H$.

The partitions $V_A$, $V_B$, $V_C$, and $V_D$ are commonly denoted by $A$, $B$, $C$, and $D$, respectively.
One can view such a partition of $V(G)$ as a labeling of $V(G)$ with the labels $A$, $B$, $C$, and $D$.
From here on, we will often omit curly braces and commas in the denotation of subsets of $\{A, B, C, D\}$.
For example, $ABD$ denotes the set $\{A, B, D\}$. 

\begin{definition}
This definition is relative to a fixed model graph $H$.
Let $L\subseteq\{A, B, C, D\}$.
We define $\Nfull{L}$ as the set of vertices of $H$ that are adjacent, through a full edge, to some vertex of~$L$.
Similarly, 
$\Ndot{L}$ is the set of vertices of $H$ that are adjacent, through a dotted edge, to some vertex of~$L$.
We say that $L$ is \emph{trivial} if $\card{L}=1$.
\qed
\end{definition}

Informally, given a model graph $H$, each label among $A$, $B$, $C$, and $D$ imposes a constraint that can be read from~$H$.
For example, for $H=\mg{19}$, the label $A$ imposes \emph{``being adjacent to $C$ and nonadjacent to $D$,''} and $C$ imposes \emph{``being adjacent to $A$ and nonadjacent to $B$.''}
A list of two or more labels imposes all the constraints of each label in the list.
For example, for $H=\mg{19}$, the list $AC$ imposes \emph{``being adjacent to $C$,  nonadjacent to $D$, adjacent to $A$, and nonadjacent to $B$.''} 
A list is called \emph{conflicting} if it imposes an unsatisfiable constraint.
For example, for $H=\mg{20}$, the list $AC$ is conflicting,
because $A$ imposes \emph{``being nonadjacent to $B$,''}
while $C$ imposes \emph{``being adjacent to $B$.''}
We will call a list \emph{non-maximal}\footnote{In~\cite{DBLP:journals/ita/DantasFGK05}, non-maximal lists are called \emph{impossible}.} if it can be extended into a longer list that imposes the same constraint.
For example, for $H=\mg{33}$, the lists $CD$ and $ACD$ impose the same constraint, namely \emph{``being adjacent to $B$, nonadjacent to $C$, and nonadjacent to $D$,'' } and therefore $CD$ is non-maximal. 

\begin{definition}
This definition is relative to a fixed model graph $H$.
Let $L\subseteq\{A, B, C, D\}$.
$L$ is \emph{conflicting} if $\Nfull{L}\cap \Ndot{L}\neq\emptyset$.
$L$ is \emph{non-maximal} if there exists $L'\subseteq\{A, B, C, D\}$ such that $L\varsubsetneq L'$ with $\Nfull{L}=\Nfull{L'}$ and $\Ndot{L}=\Ndot{L'}$.
\qed
\end{definition}

For example, let $H$ be model graph $\mg{13}$ of Figure~\ref{fig:list_models}. The list $AB$ is a non-maximal because for the set $ABC$, we have that $\Nfull{AB} = ABC = \Nfull{ABC}$ and $\Ndot{AB} = D = \Ndot{ABC}$. The list $AD$ is conflicting because $B\in\Ndot{AD}\cap\Nfull{AD}$.

\subsection{Finding $H$-Partitions}

Dantas et al.~\cite{DBLP:journals/ita/DantasFGK05} have shown that $\hph{H}$ is in polynomial time for all model graphs $H$ in Figure~\ref{fig:list_models}. 
Our algorithmic approach slightly differs from theirs because, unlike~\cite{DBLP:journals/ita/DantasFGK05}, we will not use non-maximal or conflicting lists. This is possible thanks to Lemma~\ref{lem:correctness}.

\begin{definition}
Let $H$ be a model graph, and $G$ an undirected graph.
Let $x_A, x_B, x_C$, and $x_D$ be four distinct vertices of $G$. We say that the quadruplet $(x_A, x_B, x_C, x_D)$ is $H$\emph{-isomorphic} if the subgraph of $G$ induced by these vertices is a yes-instance of $\hph{H}$ for a labeling where $x_A$, $x_B$, $x_C$, and $x_D$ are respectively labeled by $A$, $B$, $C$, and $D$.
\qed
\end{definition}

Our algorithmic approach loops over all quadruplets $(x_A, x_B, x_C, x_D)$ of vertices that are $H$-isomorphic.
Each such quadruplet yields an initial partial labeling.
Given such an initial partial labeling, we test whether it can be extended into a complete labeling, i.e., into a solution.
To this end, we repeatedly pick an unlabeled vertex, and compute its possible labels.
When we say that a label $P\in\{A,B,C,D\}$ is \emph{possible} for a vertex, we mean that the model graph does not forbid us to label that vertex with $P$, given the vertices that have already been labeled.
If only one label is possible for an unlabeled vertex, we label that vertex with that label.
If no label is possible for some unlabeled vertex, we conclude that the current labeling cannot be extended to a complete labeling.
If for every unlabeled vertex at least two labels remain possible,
then we conclude that $G$ is a yes-instance for  $\hph{H}$ (except for model graphs \model{7}, \model{10} and \model{11} where some additional tests are needed).

In \cite{DBLP:journals/ita/DantasFGK05}, \emph{refined lists} are obtained by removing conflicting and non-maximal lists. 
The description on [DdFGK05, page 141] might suggest that we must be careful to only consider refined lists.
However, the following lemma implies that conflicting or non-maximal lists will not occur in our algorithmic approach. 

\begin{lem}\label{lem:correctness}
Let $H$ be a model graph, and $G$ an undirected graph.
Assume that, during the execution of the previously described algorithm for $\hph{H}$ with input $G$,  there is a vertex $v \in V(G)$ whose set of all possible labels is $L$ with $\card{L} \geqslant 2$. Then, $L$ is neither conflicting nor non-maximal.
\end{lem}
\begin{proof}
Let $v$ be a vertex of $G$.
Let $L_{v}$ with $\card{L_{v}}\geqslant2$ denote, at some point in the execution of our algorithm, the set of those labels in $\{A, B, C, D\}$ that are possible for $v$. 
Assume for the sake of contradiction that $L_{v}$ is conflicting or non-maximal.
We consider each case in turn.
\begin{itemize}
\item 
Assume that $L_{v}$ is conflicting.
Then, there exists $R,S\in L_{v}$ and $P\in\{A,B,C,D\}$ such that $(R,P)$ is a full edge and $(S,P)$ is a dotted edge of $H$. 
From our initialization step, we know that there is a vertex $x_P$ of $G$ that is labeled by $P$.
Since the label~$R$ is possible for~$v$, it must be the case that $v$ is adjacent to $x_P$.
Since the label~$S$ is possible for $v$, it must be the case that $v$ is nonadjacent to $x_P$.
Therefore, $v$ is both adjacent and nonadjacent to $x_P$, a contradiction.
\item 
Assume that $L_{v}$ is non-maximal.
Then there exists $L'  \subseteq \{A, B, C, D\}$ such that $L_{v} \varsubsetneq L'$, $\Nfull{L_{v}} = \Nfull{L'}$ and $\Ndot{L_{v}} = \Ndot{L'}$. 
Since $L_{v} \varsubsetneq L'$, we can assume a label $P$ in $L'$ that is not in $L_{v}$.
We obtain the desired contradiction by showing that $P$ is a possible label for $v$, i.e., that $P\in L_{v}$ .
Assume towards another contradiction that $P$ is not possible for~$v$.
Then one of the following occurs:
\begin{enumerate}[(a)]
\item\label{it:full}
the model graph contains a full edge $(P,Q)$ and $v$ is nonadjacent to some vertex labeled $Q$; or
\item\label{it:dotted}
the model graph contains a dotted edge $(P,Q)$ and $v$ is adjacent to some vertex labeled $Q$.
\end{enumerate}
Assume~\eqref{it:full} (the case~\eqref{it:dotted} is symmetrical).
From $P\in L'$ and  $\Nfull{L'}=\Nfull{L_{v}}$, it follows that there exists $R\in L_{v}$ such that $(R,Q)$ is a full edge of the model graph. 
Since $R$ is possible for $v$, it follows that $v$ is adjacent to every vertex labeled $Q$, contradicting~\eqref{it:full}.
\end{itemize}
We have thus showed that in our algorithm, the set of possible labels for a vertex $v$ will never be conflicting or non-maximal.
\end{proof}

In Section~\ref{sec:datalog_strat}, we show how to encode our algorithmic approach in Datalog with stratified negation. However, we will first provide a straightforward guess-and-check program for $\hp$.

\section{Guess-and-Check Program for $\hp$}\label{sec:asp}

In all logic programs of this paper, the constants are in lowercase and the variables are in uppercase to match with Clingo syntax.
We use the binary predicate \texttt{e} to store the edges of the input graph $G$,
while the unary predicate \texttt{vertex} stores vertices.
The binary predicates \texttt{full} and  \texttt{dotted} encode, respectively, full and dotted edges of the model graph.
The labels are provided to the program as \texttt{partition(a)}, \texttt{partition(b)}, \texttt{partition(c)}, and \texttt{partition(d)}.

The following program follows the \emph{generate-and-test} (or \emph{guess-and-check}) methodology that is classical for problems in $\NP$: generate a candidate solution, and test whether it is a true solution, i.e., whether it satisfies all constraints.
In the following program, the generating part is the rule   
\texttt{1 $\lbrace$ placedIn(X,P) : partition(P) $\rbrace$ 1 :- vertex(X).}, which places every vertex in exactly one partition.
The testing part imposes that every partition must contain at least one vertex, and that the constraints of the model graph must be satisfied.

\begin{datalogpgm}
% Every vertex goes in exactly one partition.
1 { placedIn(X,P) : partition(P) } 1 :- vertex(X).

% No partition is empty.
filled(P) :- placedIn(X,P).
:- partition(P), not filled(P).  

% Constraints of the model graph.
:- placedIn(X,P), placedIn(Y,Q), full(P,Q), not adjacent(X,Y).
:- placedIn(X,P), placedIn(Y,Q), dotted(P,Q), adjacent(X,Y).
\end{datalogpgm}

\section{Datalog Programs for $\hp$}\label{sec:datalog_strat}

In this section, we show how to encode $\hph{H}$ in Datalog with stratified negation, for model graphs $H$ in Figure~\ref{fig:list_models}.
Section~\ref{sec:general} encodes the approach of~\cite{DBLP:journals/ita/DantasFGK05} in general.
We argue that this general encoding is correct for every model graph in Figure~\ref{fig:list_models} that is distinct from $\mg{7}$, $\mg{10}$, and $\mg{11}$.
We then show how to adjust the general encoding for these three model graphs.
Finally, we show that for model graphs $H$ with an isolated vertex,
non-recursive Datalog with negation suffices to check the existence of $H$-partitions.


\subsection{General Datalog Program}\label{sec:general}

Dantas et~al.~\cite{DBLP:journals/ita/DantasFGK05} showed that $\hph{H}$ can be solved in polynomial time for all model graphs  $H$ in Figure~\ref{fig:list_models}.
Following their algorithm and using Lemma~\ref{lem:correctness}, we check whether some initial valid labeling of four vertices, called a \textit{base}, can be extended to a complete labeling of all vertices.

\paragraph{Finding a base.} The algorithm loops over all quadruplets of distinct vertices and checks whether a quadruplet is $H$-isomorphic.
Our Datalog rules are as follows:

\begin{datalogpgm}
base(X,Y,Z,T) :- distinct(X,Y,Z,T), not problematic_base(X,Y,Z,T).
\end{datalogpgm}
Here, the predicate \texttt{distinct(X,Y,Z,T)} checks whether $X$, $Y$, $Z$ and $T$ are four distinct vertices of~$G$, and \texttt{problematic\_base} is defined, for each possible edge in the model graph~$H$, with two rules. For example, if we consider the edge $(a,b)$, the rules are:

\begin{datalogpgm}
problematic_base(X,Y,Z,T) :- dotted(a,b), e(X,Y), vert(Z), vert(T).
problematic_base(X,Y,Z,T) :- full(a,b), not e(X,Y), vert(X), vert(Y), vert(Z), 
                             vert(T).
\end{datalogpgm}
Symmetrical are added for edges $(a,c)$, $(a, d)$, $(b,c)$, $(b, d)$, and $(c, d)$.

\paragraph{Extension to a complete labeling.}
First, for every base $(I,J,K,L)$ of four vertices,
we label $I$ with $A$, $J$ with $B$, $K$ with $C$, and $L$ with $D$.

\begin{datalogpgm}
inPart(I,a,I,J,K,L) :- base(I,J,K,L).
inPart(J,b,I,J,K,L) :- base(I,J,K,L).
inPart(K,c,I,J,K,L) :- base(I,J,K,L).
inPart(L,d,I,J,K,L) :- base(I,J,K,L).
\end{datalogpgm}

We keep track that a vertex that has been labeled in a base must not be re-labeled later on.
\begin{datalogpgm}
done(I,I,J,K,L) :- base(I,J,K,L).
done(J,I,J,K,L) :- base(I,J,K,L).
done(K,I,J,K,L) :- base(I,J,K,L).
done(L,I,J,K,L) :- base(I,J,K,L).
\end{datalogpgm}

Next, rather than computing the set of possible labels for a yet unlabeled vertex~$X$, we use the predicate \texttt{imp(X,P,I,J,K,L)} with the meaning that the vertex  $X$ cannot be labeled by~$P\in\{A, B, C, D\}$, relative to the fixed base $(I,J,K,L)$.
The Datalog rules are straightforward.
The second rule, for example, expresses that for a full edge $(P,Q)$ in the model graph, if a vertex $Y$ has been labeled with~$Q$,
then a vertex $X$ that is nonadjacent to $Y$ cannot be labeled with~$P$. 
One should read ``prob'' as a shorthand for ``problematic.''

\begin{datalogpgm}
imp(X,P,I,J,K,L) :- vert(X), full_prob(X,P,I,J,K,L), not done(X,I,J,K,L).
full_prob(X,P,I,J,K,L) :- full(P,Q), inPart(Y,Q,I,J,K,L), not e(X,Y), vert(X).

imp(X,P,I,J,K,L) :- vert(X), dot_prob(X,P,I,J,K,L), not done(X,I,J,K,L).
dot_prob(X,P,I,J,K,L) :- dotted(P,Q), inPart(Y,Q,I,J,K,L), e(X,Y).
\end{datalogpgm}

For every vertex $X$ of $G$ that is not in the fixed base $(I,J,K,L)$, we now consider two possibilities:
\begin{enumerate}[(1)]
\item 
if three labels among $A$, $B$, $C$, and $D$ are impossible for $X$, then $X$ is labeled with the remaining fourth label; and
\item 
if all labels among $A$, $B$, $C$, and $D$ are impossible for $X$, then the base cannot be extended to a complete labeling.
\end{enumerate}
Finally, if some base never encounters the second case, the answer to $\hph{H}$ is ``yes''; otherwise the answer is ``no''.
Note incidentally that in a ``yes''-instance, there may be vertices~$X$ that are not labeled by the first case.

\begin{datalogpgm}
% First case.
inPart(X,a,I,J,K,L) :- imp(X,b,I,J,K,L), imp(X,c,I,J,K,L), imp(X,d,I,J,K,L).
inPart(X,b,I,J,K,L) :- imp(X,a,I,J,K,L), imp(X,c,I,J,K,L), imp(X,d,I,J,K,L).
inPart(X,c,I,J,K,L) :- imp(X,a,I,J,K,L), imp(X,b,I,J,K,L), imp(X,d,I,J,K,L).
inPart(X,d,I,J,K,L) :- imp(X,a,I,J,K,L), imp(X,b,I,J,K,L), imp(X,c,I,J,K,L).

% Second case.
bad_init(I,J,K,L) :- vert(X), base(I,J,K,L), imp(X,a,I,J,K,L), 
                     imp(X,b,I,J,K,L), imp(X,c,I,J,K,L), imp(X,d,I,J,K,L).

yes_instance() :- base(I,J,K,L), not bad_init(I,J,K,L).
no_instance() :- not yes_instance().
\end{datalogpgm}

We have the following result.

\begin{thm}\label{the:correct}
If $H$ is one of the model graphs of Figure~\ref{fig:list_models} such that $H$ is distinct from model graphs~$\mg{7}$, $\mg{10}$, and~$\mg{11}$, then $\hph{H}$ is expressible in Datalog with stratified negation.
\end{thm}

\begin{proof}[Proof sketch.]
Correctness follows from \cite{DBLP:journals/ita/DantasFGK05} and Lemma~\ref{lem:correctness}. 
The precedence graph of the program is given in Figure~\ref{fig:PDG_conj_2_b}.
\end{proof}

\begin{figure}
\begin{center}
\includegraphics[scale=0.4]{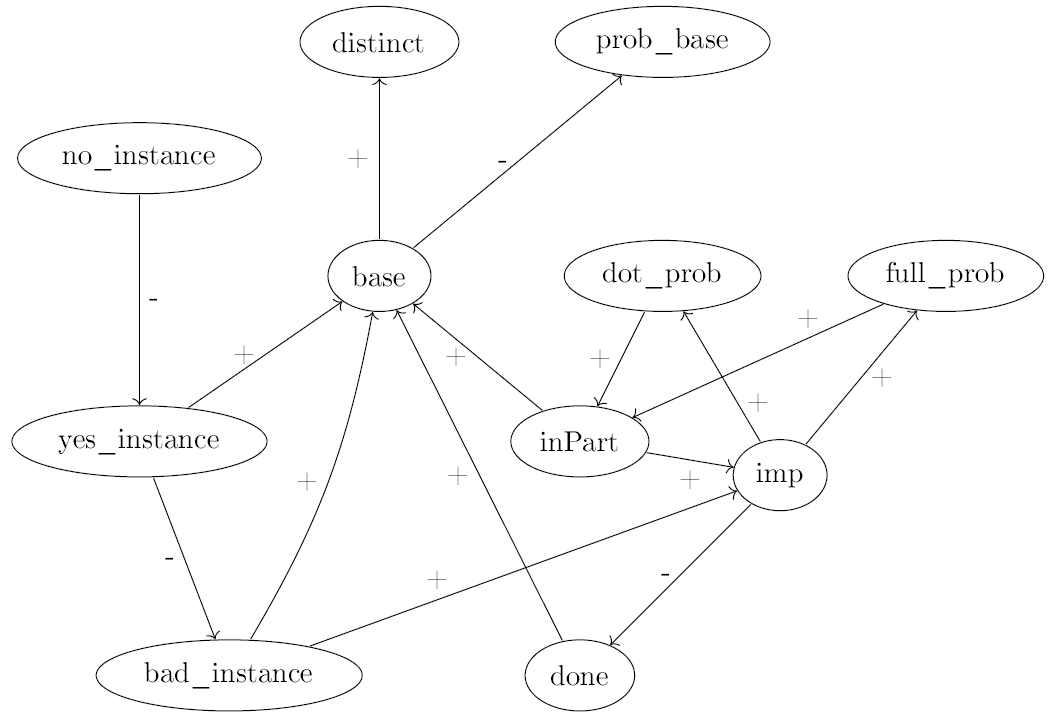}
\caption{Precedence Graph of the Datalog program.}\label{fig:PDG_conj_2_b}
\end{center}
\end{figure}

\subsection{Adjustment for Model Graph $\mg{7}$}

In~\cite{DBLP:journals/ita/DantasFGK05}, the model graph $\mg{7}$ of Figure~\ref{fig:list_models} is reduced to the model graph $K_2$. 
This is possible because the constraints for $A$ and $C$ are identical (i.e., \emph{being adjacent to both $B$ and $D$}), and the constraints for $B$ and $D$ are identical.
Then, every non-base vertex labeled by~$C$ can also be labeled by~$A$,
and every non-base vertex labeled by~$D$ can also be labeled by~$B$.
Our previous Datalog program would fail, however, for the reason that the set of possible labels for a  non-base vertex is never a singleton.

For example, consider the graph $G$ in Figure~\ref{fig:ex_model_7} with the base $(1, 2, 3, 4)$. The list of possible labels for the vertex $5$ is $AC$, and the list of possible labels for the vertex $6$ is $BD$. Our previous algorithm would terminate here and conclude that $G$ is a yes-instance. 
Note, however, that if vertex $5$ is labeled with $A$ (the case where we label~$5$ with~$C$ is symmetrical), it is impossible to label the vertex~$6$ with~$B$ or~$D$, because vertices~$5$ and~$6$ are non-adjacent. 

\begin{figure}\centering
\includegraphics[scale=0.3]{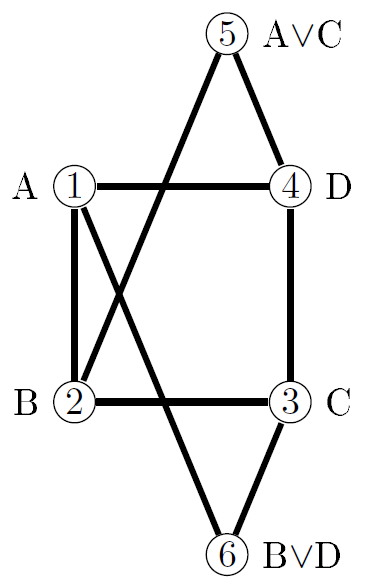}
\caption{The model graph $\mg{7}$ requires an additional test.}
\label{fig:ex_model_7}
\end{figure}

We adjust our Datalog program as follows: if the set of possible labels of a vertex~$X$ has size~ $2$, the new program will label $X$ with one of these two possibilities.
The rest of the program remains unchanged. 
That is, in the box that immediately precedes Theorem~\ref{the:correct}, we replace the four rules for the predicate \texttt{inPart} by the following two rules:

\begin{datalogpgm}
inPart(X,a,I,J,K,L) :- imp(X,b,I,J,K,L), imp(X,d,I,J,K,L).
inPart(X,b,I,J,K,L) :- imp(X,a,I,J,K,L), imp(X,c,I,J,K,L).
\end{datalogpgm}

\subsection{Adjustment for Model Graphs $\mg{10}$ and $\mg{11}$}

As in~\cite{DBLP:journals/ita/DantasFGK05}, the model graphs $\mg{10}$ and $\mg{11}$ of Figure~\ref{fig:list_models} need a special treatment.
We illustrate the issue by the model graph $\mg{10}$ and the graph $G$ given in Figure~\ref{fig:ex_model_10}~\emph{(right)}.
For the base $(1,2,3,4)$, the list of possible labels for the vertex $5$ is $BC$, and the list of possible labels for the vertex $6$ is $AD$. So our original algorithm would terminate and wrongly conclude that the base $(1,2,3,4)$ can be extended to a complete labeling. 
Nevertheless, we claim that this base cannot be extended to a complete labeling. Indeed, let us label the vertex~$5$ by~$B$ (the case where we label $5$ with $C$ is symmetrical).
Since $B$ is adjacent, through full edges, to both $A$ and $D$,
and since vertices~$5$ and~$6$ are nonadjacent,
we cannot label~$6$ with $A$ or $D$.
We note incidentally that the graph $G$ is a yes-instance through other bases, for example, $(1,2,3,6)$.

\begin{center}
\begin{figure}
\begin{center}
\setlength{\tabcolsep}{5\tabcolsep}
\begin{tabular}{cc}
\includegraphics[scale=0.3]{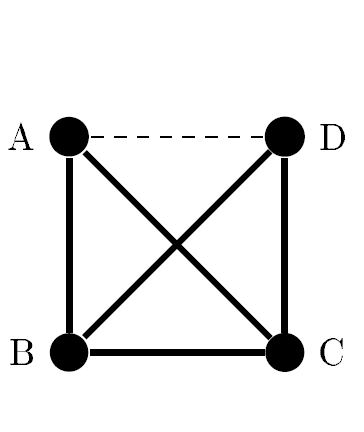}
&
\includegraphics[scale=0.3]{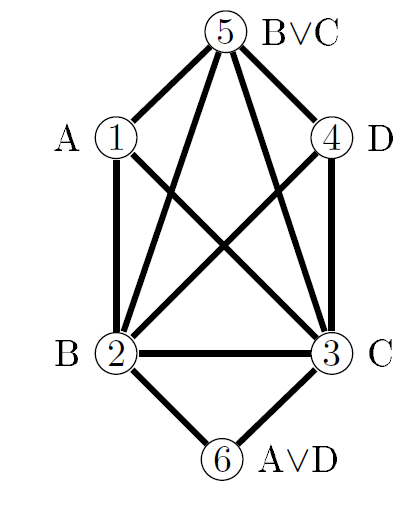}\\
$\mg{10}$ & \\
\end{tabular}
\end{center}
\caption{The model graph $\mg{10}$ requires an additional test.}
\label{fig:ex_model_10}
\end{figure}
\end{center}

It can be verified that for the model $\mg{10}$, the lists of possible labels are $AD$ and $BC$;  for the model $\mg{11}$, the lists of possible labels are $AC$ and $BD$.
We have the following result.

\begin{lem} (Cf.\ \cite{DBLP:journals/ita/DantasFGK05}) \label{lem:put_in_A}
Let $H$ be one of the model graphs~$\mg{10}$ or $\mg{11}$ in Figure~\ref{fig:list_models}.
If there exists a solution to $\hph{H}$ with a base $(i,j,k,l)$, then there is also a solution with the same base such that:
\begin{itemize}
\item 
each vertex having $A$ in its list of possible labels is labeled $A$; and
\item 
each vertex having $B$ in its list of possible labels is labeled $B$.
\end{itemize}
\end{lem}

We now use Lemma~\ref{lem:put_in_A} to adjust our previous Datalog program of Section~\ref{sec:general}.  
First, we copy the labeled vertices.

\begin{datalogpgm}
label(X,a,I,J,K,L) :- inPart(X,a,I,J,K,L).
label(X,b,I,J,K,L) :- inPart(X,b,I,J,K,L).
label(X,c,I,J,K,L) :- inPart(X,c,I,J,K,L).
label(X,d,I,J,K,L) :- inPart(X,d,I,J,K,L).
\end{datalogpgm}

Next, if $B$ is an impossible label for a vertex but $A$ is not, the vertex is labeled by $A$. Symmetrically, if $A$ is an impossible label for a vertex but $B$ is not, the vertex is labeled by $B$.

\begin{datalogpgm}
label(X,a,I,J,K,L) :- not imp(X,a,I,J,K,L), imp(X,b,I,J,K,L).
label(X,b,I,J,K,L) :- imp(X,a,I,J,K,L), not imp(X,b,I,J,K,L).
\end{datalogpgm}

We now have to check that this labeling is correct according to the model graph. If not, we reject the initial base.

\begin{datalogpgm}
problem(X,P,I,J,K,L) :- label(X,P,I,J,K,L), full(P,Q), label(Y,Q,I,J,K,L), 
                        not e(X,Y).
problem(X,P,I,J,K,L) :- label(X,P,I,J,K,L), dotted(P,Q), label(Y,Q,I,J,K,L), 
                        e(X,Y).
bad_base(I,J,K,L) :- problem(X,P,I,J,K,L).
\end{datalogpgm}

Finally, the definition of \texttt{yes\_instance()} is changed as follows:

\begin{datalogpgm}
yes_instance() :- base(I,J,K,L), not bad_base(I,J,K,L), not bad_init(I,J,K,L).
no_instance() :- not yes_instance().
\end{datalogpgm}

Along the lines of Theorem~\ref{the:correct}, we obtain the following result.

\begin{thm}
If $H$ is one of the model graphs $\mg{10}$ or $\mg{11}$ in Figure~\ref{fig:list_models}, then $\hph{H}$ is expressible in Datalog with stratified negation.
\end{thm}


\subsection{Model Graphs with Isolated Vertices}\label{ssec:isolated}

A vertex in a model graph $H$ is \emph{isolated} if it is not adjacent to a full or dotted edge.
The model graphs in Figure~\ref{fig:list_models} with at least one isolated vertex are $\mg{1}$, $\mg{2}$, $\mg{3}$, $\mg{4}$, $\mg{22}$, and $\mg{29}$.
It can be easily seen that for these model graphs, $\hph{H}$ can be solved in non-recursive Datalog with negation (i.e., in predicate logic).

\begin{lem} \label{lem:isolated_vertex}
Let $G$ be an undirected graph. 
Let $H$ be a model graph with an isolated vertex.
Then the following are equivalent:
\begin{enumerate}
\item\label{it:quadruplet}
$G$ contains a $H$-isomorphic quadruplet of distinct vertices; and
\item \label{it:yesinstance}
$G$ is a yes-instance of the problem $\hph{H}$.
\end{enumerate}
\end{lem}
\begin{proof}
\framebox{\ref{it:quadruplet}$\implies$\ref{it:yesinstance}}
Let $P$ be an isolated vertex of $H$ and $(i,j,k,l)$ be a quadruplet of distinct vertices of $G$ which are $H$-isomorphic.
We can label the other vertices of $G$ by $P$ and we obtain a complete labeling of $G$. Indeed, as $P$ is isolated, there are only constraints including the vertices $(i, j, k, l)$, but these constraints are fulfilled because the quadruplet $(i,j,k,l)$ is $H$-isomorphic.

\framebox{\ref{it:yesinstance}$\implies$\ref{it:quadruplet}}
The other direction of the equivalence is trivial because $G$ must, at least, contain a base to be a yes-instance.
\end{proof}

From Section~\ref{sec:general}, it follows that non-recursive Datalog suffices to test the existence of a $H$-isomorphic quadruplet of distinct vertices. 
Thus, we obtain the following result.

\begin{thm}
If $H$ is a model graph with an isolated vertex, then $\hph{H}$ is expressible in non-recursive Datalog with negation.
\end{thm}

\section{Experiments} \label{sec:gen_instances}

To compare the efficiency of our programs in an existing answer set solver,
we need a generator for yes-instances and no-instances of arbitrary size.
The motivation for distinguishing between yes-and no-instances comes from the asymmetry of $\NP$-problems: on a yes-instance, an algorithm can stop as soon as the first $H$-partition is found; but on no-instances no such early stopping is possible.   
In this section, we discuss generators for yes- and no-instances, and then we report on execution times of our programs on these instances.


\subsection{Generating Yes-Instances}
Let $n$, $m$ be two positive integers, and $H$ a model graph.
We want to build a graph $G$ with $n$~vertices and $m$~edges that is a yes-instance for the problem $\hph{H}$.
Necessarily, $n\geqslant4$ and $m\leq\frac{n(n-1)}{2}$, where the latter is the number of edges in the complete graph with~$n$ vertices.

We denote $\INpos=\{1,2,3,\ldots\}$.
Let $a ,b, c, d \in \INpos$ such that $a+b+c+d=n$. One can wonder if there exists a yes-instance with $m$ edges where the numbers of vertices labeled by $A$, $B$, $C$, and $D$ are, respectively, $a$, $b$, $c$, and $d$.

\begin{definition}
We define $m_{\mini}^H(a,b,c,d)$ as the smallest number $k$ such that some graph~$G$ with $\card{E(G)}=k$ admits a $H$-partition $(V_A, V_B, V_C, V_D)$ such that $\card{V_A} = a$, $\card{V_B} = b$, $\card{V_C} = c$, and $\card{V_D} = d$.
Symmetrically, we define $m_{\maxi}^H(a,b,c,d)$ as the largest number $k$ such that some graph~$G$ with $\card{E(G)} = k$ admits a $H$-partition $(V_A, V_B, V_C, V_D)$ such that $\card{V_A} = a$, $\card{V_B} = b$, $\card{V_C} = c$, and $\card{V_D} = d$.
\qed
\end{definition}

We write $p$ for a label and also for the number of vertices labeled with this label. If the context is not clear, we write $\#p$ for the number of vertices labeled by $p$.
We have the following result.


\begin{lem}\label{lem:min_max}
Let $H$ be a model graph.
Let $a, b, c, d \in \INpos$ such that $a+b+c+d=n$.
Then,
\begin{align}
m_{\mini}^H(a, b, c, d) &= \sum_{(p, q) \in \Full{H}} \#p \cdot \#q\label{eq:min_max_one}\\
m_{\maxi}^H(a, b, c, d) &= \frac{n (n-1)}{2} - \sum_{(p, q) \in \Dotted{H}} \#p \cdot \#q
\label{eq:min_max_two}
\end{align}
where $\Full{H}$ is the set of the full edges of $H$, and $\Dotted{H}$ is the set of the dotted edges of~$H$.
\end{lem}

\begin{proof}
We  give a proof of~\eqref{eq:min_max_one}.
We first show that there exists an undirected graph $G_{\mini}$ with $\sum_{(p,q) \in \Full{H}} \#p \cdot \#q$ edges and $n$~vertices that has a solution where the numbers of vertices labeled by $A$, $B$, $C$, and $D$ are, respectively, $a$, $b$, $c$, and $d$.
Let $G_{\mini}$ be an undirected graph whose vertices are 
$$V(G_{\mini}) = \{A_1, A_2, \dots, A_a, B_1, \dots, B_b, C_1, \dots,  C_c, D_1, \dots, D_d\}$$
 and the edges of $G_{\mini}$ are defined as follows: for $p_i, q_j \in V(G_{\mini})$, $(p_i, q_j)$ is an edge of $G_{\mini}$ if and only if $(p, q)$ is a full edge of $H$.

Since for each full edge $(p,q)$ in $H$, each vertex labeled with $p$ is connected to each vertex labeled with $q$, there are $\#p \cdot \#q$ such edges. Moreover, there are no other edges than those described here.
As we sum up these values for each full edge $(p,q)$, we obtain the desired number of edges in $G_{\mini}$.

We claim that this labeling satisfies all the constraints.
The full constraints are satisfied by construction. 
The dotted constraints are also satisfied because $G_{\mini}$ contains only edges connecting two vertices whose labels are adjacent through a full edge in $H$.

It remains to be shown that there exists no graph with strictly less than $\sum_{(p, q) \in \Full{H}} \#p \cdot \#q$  edges that is a yes-instance where the number of vertices labeled by $A$, $B$, $C$, and $D$ are, respectively, $a$, $b$, $c$, and $d$.
Assume for the sake of contradiction that there exists such a graph~$G'$. We can assume without loss of generality that $G'$ and $G_{\mini}$ have the same set of vertices.
Since $G'$ has strictly less edges than $G_{\mini}$, there is an edge $(u,v)$ in $G_{\mini}$ that is not in $G'$. 
Since $G'$ is a yes-instance, it has a solution.
Let $p$ be the label of $u$ in this solution, and $q$ the label of $v$.
Since $(u,v)$ is an edge of $G_{\mini}$, we have that $(p,q)$ is a full edge of $H$. So, we have two vertices $u$ and $v$ that are labeled by $p$ and $q$ respectively, but there is no edge in $G'$ between $u$ and $v$. This is a contradiction with the fact that $G'$ is a yes-instance.
This concludes the proof of ~\eqref{eq:min_max_one}. 

The proof of \eqref{eq:min_max_two} is similar because it suffices to consider $G_{\maxi}$ which is obtained from a complete graph with $n$ edges by removing all the edges between two vertices whose labels are connected by a dotted edge in $H$.
\end{proof}

For example, let us fix $a, b, c, d \in \INpos$ with the model graph $(23)$ of Figure~\ref{fig:list_models} and $n = a+b+c+d$.
According to the Figure~\ref{fig:model_23}, we have:
$$m_{\mini}^{(23)}(a, b, c, d) = ab + bc$$
$$m_{\maxi}^{(23)}(a, b, c, d) = \frac{n (n-1)}{2} - ac - bd - cd$$
Moreover, we claim that, for all $m$ such that $m_{\mini}^{(23)}(a, b, c, d) \leqslant m \leqslant m_{\maxi}^{(23)}(a, b, c, d)$, there exists a graph $G$ with $n$ vertices and $m$ edges such that $G$ is a yes-instance for a labeling where the numbers of vertices labeled by $A$, $B$, $C$, and $D$ are, respectively, $a$, $b$, $c$, and $d$. The idea is that we can extend the graph for $m_{\mini}^{(23)}(a, b, c, d)$ by adding edges until we reach $m$ edges such that every added edge does not violate any dotted constraint.
We will now formalize this construction.

\begin{figure}
\begin{center}
\includegraphics[scale=0.5]{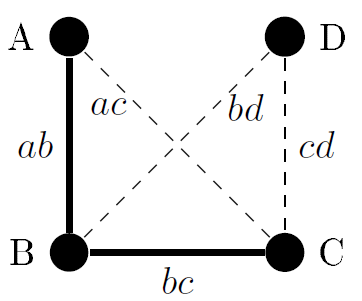}
\caption{Model graph $\mg{23}$.}\label{fig:model_23}
\end{center}
\end{figure}

\begin{definition}
Let $G$ and $G'$ be two undirected graphs,
We say that $G$ is an \emph{extension} of $G'$, denoted $G' \subseteq G$, if $V(G') = V(G)$ and $E(G') \subseteq E(G)$. If $G' \subseteq G$, we also say that $G'$ is \emph{extended by} $G$.
\qed
\end{definition}

In other words, $G$ is an extension of $G'$ if it is possible to obtain $G$ from $G'$ by adding some edges.
One can easily see that $G_{\mini}\subseteq G_{\maxi}$. Indeed, we take $G_{\mini}$ and consider the labeling induced by the construction of $G_{\mini}$. We can add all possibles edges that do not violate any dotted constraint. In this way, we obtain $G_{\maxi}$ by definition.

\begin{thm}\label{the:sandwich}
Let $H$ be a model graph.
Let $a, b, c, d \in \INpos$ and $n = a+b+c+d$.
Let $G_{\mini}$ and $G_{\maxi}$ be the graphs with $n$ vertices and respectively $m_{\mini}^{H}(a, b, c, d)$ and $m_{\maxi}^{H}(a, b, c, d)$ edges defined in (the proof of) Lemma~\ref{lem:min_max}.
Every graph $G$ such that $G_{\mini} \subseteq G \subseteq G_{\maxi}$ has a solution $(V_A, V_B, V_C, V_D)$ where $\card{V_A} = a$, $\card{V_B} = b$, $\card{V_C} = c$, and $\card{V_D} = d$.
\end{thm}
\begin{proof}
Let us consider the same labeling of the vertices in $G$ as in $G_{\mini}$. So the number of vertices labeled by $A$, $B$, $C$, and $D$ are, respectively, $a$, $b$, $c$, and $d$.
It remains to be shown that this labeling satisfies all the constraints.

\begin{itemize}
\item 
Let $(p,q)$ be a full edge of $H$. 
We have to show that all the vertices labeled by $p$ are adjacent to all vertices labeled by $q$. This holds true because $G$ is an extension of $G_{\mini}$ and $G_{\mini}$ already satisfies this constraint.
\item 
Let $(p,q)$ be a dotted edge of $H$.
We have to show that all the vertices labeled by $p$ are nonadjacent to all vertices labeled by $q$. This holds true because $G_{\maxi}$ is an extension of $G$ and $G_{\maxi}$ already satisfies this constraint.
\end{itemize}
All the constraints are satisfied and, consequently, $G$ is a yes-instance. 
\end{proof}

This theorem has a natural corollary.

\begin{corol} \label{cor:m}
Let $H$ be a model graph, and $a, b, c, d \in \INpos$ such that $n=a+b+c+d$.
Let $m$ be an integer such that $m_{\mini}^{H}(a, b, c, d) \leqslant m \leqslant m_{\maxi}^{H}(a, b, c, d)$.
There is an undirected graph $G$ with $n$ vertices and $m$ edges such that $G$ has a solution where the numbers of vertices labeled by $A$, $B$, $C$, and $D$ are, respectively, $a$, $b$, $c$, and $d$.
\end{corol}
\begin{proof}
We have already showed in the proof of Lemma \ref{lem:min_max} that $G_{\mini}$ and $G_{\maxi}$ have respectively $m_{\mini}^{H}(a, b, c, d)$ and $m_{\maxi}^{H}(a, b, c, d)$ edges. So, we just have to consider a graph $G$ with $m$ vertices such that $G_{\mini} \subseteq G \subseteq G_{\maxi}$ and use Theorem~\ref{the:sandwich}.
\end{proof}

Based on Corollary~\ref{cor:m}, we introduce the following generator of yes-instances.

\paragraph{Generator for yes-instances}

\begin{description}
\item[Input:] A model graph $H$; two positive integers $n$ and $m$ such that $n\geqslant4$.
\item[Output:] Yes-instances with $n$ vertices and $m$ edges.
\item[Generator:]
For every quadruplet $(a,b,c,d)$ of positive integers such that $a+b+c+d=n$ and  $m_{\mini}^{H}(a, b, c, d) \leqslant m \leqslant m_{\maxi}^{H}(a, b, c, d)$,
\begin{quote}
build a yes-instance from $G_{\mini}$ by adding edges that do not violate any dotted constraint, until the number of edges reaches $m$;
\end{quote}
if no such quadruplet exists, return \emph{``there is no such yes-instance.''}
\end{description}

\subsection{Generating No-Instances}

\begin{figure*}
\begin{center}
\setlength{\tabcolsep}{3\tabcolsep}
\begin{tabular}{ccc}
\includegraphics[scale=0.3]{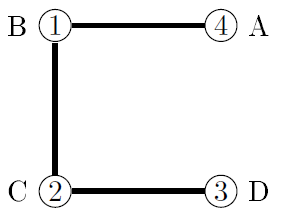}
&
\includegraphics[scale=0.3]{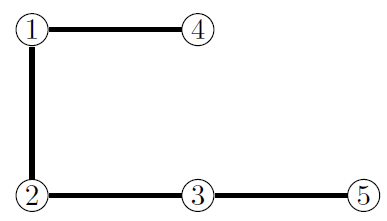}
&
\includegraphics[scale=0.3]{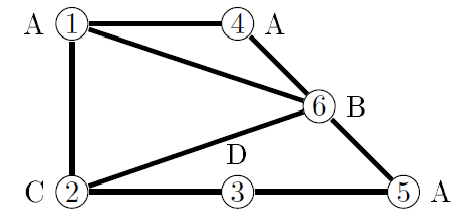}
\\
$G_{1}$ & $G_{2}$ & $G_{3}$
\end{tabular}
\end{center}
\caption{
For the model graph $\mg{6}$, we have that
$G_{1}$ is a yes-instance, $G_{2}$ is a no-instance, and $G_{3}$ is a yes-instance.}\label{fig:ex_no_instance}
\end{figure*}

Given an integer $n\geqslant 4$ and a model graph $H$, we want to build a graph $G$ with $n$ vertices that is a no-instance for the problem $\hph{H}$.
The following proposition tells us when this is possible.

\begin{prop}\label{prop:existence_no_inst}
Let $H$ be a model graph and $n$ an integer such that $n\geqslant 4$.
There is an undirected  graph~$G$ with $n$ vertices that is a no-instance for $\hph{H}$ if and only if $H$ is not the model $\mg{1}$ in Figure~\ref{fig:list_models}.
\end{prop}
\begin{proof}
Assume that $H$ is the model graph $\mg{1}$. Then, any graph $G$ with $n$ vertices is a yes-instance because there is no constraint to satisfy, due to the absence of full and dotted edges.

For the opposite direction, assume that $H$ is not the model graph $\mg{1}$.
Then $H$ contains a dotted edge or a full edge.
\begin{itemize}
\item If $H$ contains at least one full edge, the empty graph with $n$ vertices and no edge is a no-instance because the full constraint will never be satisfied for any labeling of the vertices. 
\item 
If $H$ contains at least one dotted edge, the complete graph with $n$ vertices and $\frac{n(n-1)}{2}$ edges is a no-instance because the dotted constraint will never be satisfied for any labeling of the vertices.
\end{itemize}
This concludes the proof.
\end{proof}

The undirected graphs in the proof of Proposition~\ref{prop:existence_no_inst} are of little interest in an experimental validation, because they do not admit a $H$-isomorphic quadruplet of vertices, and therefore the Datalog program would not even enter its recursive part.
%
%
Instead, for a model graph $H$, we would like to generate no-instances that contain a $H$-isomorphic quadruplet. Note here that, by Lemma \ref{lem:isolated_vertex}, such no-instances do not exist if $H$ has an isolated vertex.

It would be interesting to develop an approach similar to that for yes-instances:
start with a small no-instance and extend it until some desired number of vertices  is reached. Nevertheless, this seems a difficult task, because the addition of edges for new vertices can turn no-instances into yes-instances, and vice versa.
For example, in Figure~\ref{fig:ex_no_instance}, $G_{1}$ is extended by $G_{2}$, and $G_{2}$ is extended by $G_{3}$. 
For the model graph $\mg{6}$, we have that $G_{2}$ is a no-instance, while $G_{1}$ and $G_{3}$ are yes-instances.

In view of the aforementioned difficulties, we have decided to generate no-instances in a randomized fashion.
Given $n \geqslant 4$ and model graph $H$ that is not the model $\mg{1}$ of Figure~\ref{fig:list_models}, the idea is to generate random graphs with $n$ vertices until a no-instance for $\hph{H}$ is found.
To see whether this approach is viable, we have generated random graphs with $n=100$ by adding an edge to a graph with a given probability, its expected density. The expected density is considered as a random variable whose probability distribution is a Gaussian distribution, $\mathcal{N}(0.5, 0.25)$. 
This distribution seems better than a ``classical'' uniform distribution on the expected density because there are more possible graphs with a density of $0.5$ than $0.01$ or $0.99$.
Table \ref{tab:proportion_yes_no} shows, for each model graph in Figure~\ref{fig:list_models}, the number of yes-instances and no-instances over a sample of $1000$ generated graphs. 

For the model graphs $\mg{1}$, $\mg{2}$, $\mg{3}$, $\mg{4}$, $\mg{18}$, $\mg{22}$, $\mg{29}$ of Figure~\ref{fig:list_models}, there are more yes-instances than no-instances. Except for $\mg{18}$, these model graphs have an isolated vertex and, by Lemma~\ref{lem:isolated_vertex}, an input graph is a yes-instance as soon as it contains a $H$-isomorphic quadruplet. For the model graph $\mg{18}$, a graph is easily a yes-instance because it suffices to have an edge $(u,v)$ such that $u$ is nonadjacent to at least two other vertices: the vertex $u$ will be labeled by $C$, all vertices adajacent to~$u$ will have the label $A$, and all nonadjacent vertices will have label $B$ or $D$.
For model graphs that are distinct from the above seven model graphs, there are considerably more no-instances than yes-instances.
In conclusion, our approach should find a no-instance in a few tries, except for the seven models that have been discussed before, for which we will use the no-instances of the proof of Proposition~\ref{prop:existence_no_inst}.
 
\begin{table}\centering
\begin{tabular}{|c|r|r||c|r|r|}
\hline
model graph & \multicolumn{1}{c|}{\# yes} & \multicolumn{1}{c||}{\# no} & model graph & \multicolumn{1}{c|}{\# yes} & \multicolumn{1}{c|}{\# no}\\
\hline\hline
$\mg{1}$ & $1000$ & $0$ & $\mg{18}$ & $999$ & $1$ \\
\hline
$\mg{2}$ & $1000$ & $0$ & $\mg{19}$ & $185$ & $815$ \\
\hline
$\mg{3}$ & $1000$ & $0$ & $\mg{20}$ & $15$ & $985$ \\
\hline
$\mg{4}$ & $998$ & $2$ & $\mg{21}$ & $35$ & $965$ \\
\hline
$\mg{5}$ & $15$ & $985$ & $\mg{22}$ & $1000$ & $0$ \\
\hline
$\mg{6}$ & $175$ & $825$ & $\mg{23}$ & $15$ & $985$ \\
\hline
$\mg{7}$ & $13$ & $987$ & $\mg{24}$ & $15$ & $985$ \\
\hline
$\mg{8}$ & $15$ & $985$ & $\mg{25}$ & $14$ & $986$ \\
\hline
$\mg{9}$ & $13$ & $987$ & $\mg{26}$ & $56$ & $944$ \\
\hline
$\mg{10}$ & $9$ & $991$ & $\mg{27}$ & $15$ & $985$ \\
\hline
$\mg{11}$ & $0$ & $1000$ & $\mg{28}$ & $5$ & $995$ \\
\hline
$\mg{12}$ & $0$ & $1000$ & $\mg{29}$ & $1000$ & $0$ \\
\hline
$\mg{13}$ & $0$ & $1000$ & $\mg{30}$ & $0$ & $1000$ \\
\hline
$\mg{14}$ & $2$ & $998$ & $\mg{31}$ & $0$ & $1000$ \\
\hline
$\mg{15}$ & $4$ & $996$ & $\mg{32}$ & $6$ & $994$ \\
\hline
$\mg{16}$ & $4$ & $996$ & $\mg{33}$ & $15$ & $985$ \\
\hline
$\mg{17}$ & $183$ & $817$ & $\mg{34}$ & $12$ & $988$ \\
\hline
\end{tabular}
\caption{Proportion of yes- and no-instances.}\label{tab:proportion_yes_no}
\end{table}

\begin{figure}
\begin{center}
\includegraphics[scale=1.0]{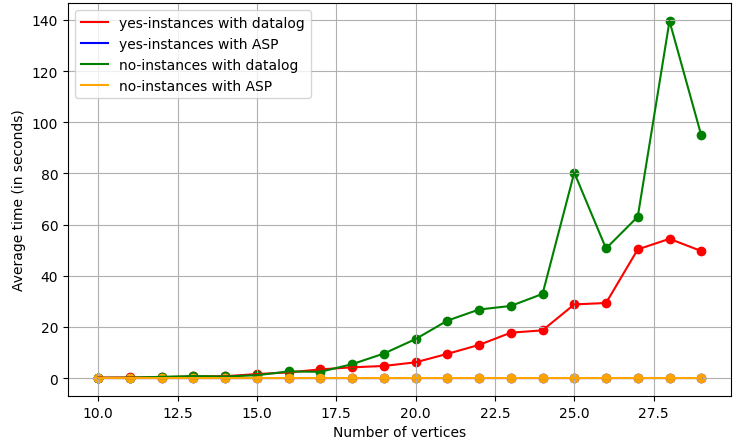}
\end{center}
\caption{Experimental results in Clingo (the blue line is hidden behind the yellow line).}\label{fig:graphic_time}
\end{figure}

\subsection{Execution Times}\label{sec:comparing}

We have used the answer set solver Clingo~\cite{DBLP:conf/lpnmr/GebserKKS11, DBLP:journals/corr/GebserKKS14} to compare the execution times of ASP guess-and-check programs with   programs in Datalog with stratified negation.
Figure~\ref{fig:graphic_time} shows execution times on input graphs of different small sizes that were generated like explained in Section~\ref{sec:gen_instances}.
The execution times reported in Figure~\ref{fig:graphic_time} are average execution times over the $34$~model graphs of Figure~\ref{fig:list_models}.
It is immediately clear from this figure that in Clingo, the guess-and-check program is much more efficient than our Datalog programs with stratified negation.
The difference in efficiency was even more striking for larger graphs.
In fact, on graphs with more than $50$ vertices,  we encountered memory problems when executing our Datalog programs in Clingo, which may be due to the size of the grounded program.

\section{Conclusion} \label{sec:conclusion}

We presented a logical programming approach to the problem of finding $H$-partitions in undirected graphs.
In particular, we showed how the polynomial-time solutions given in~\cite{DBLP:journals/ita/DantasFGK05} can be encoded in Datalog with stratified negation.
In an experimental validation, we found that in the answer set solver Clingo,
these Datalog programs execute much slower than a simple guess-and-check program.
We also studied the problem of generating yes- and no-instances for $\hp$ with a fixed number of vertices or edges. 

Our experiments indicate that when using an answer set solver, there may be no benefit in writing programs that are theoretically more efficient than guess-and-check programs.
It would be interesting to investigate execution times of our Datalog programs on a dedicated Datalog engine.
It would also be interesting to find systematic methods for generating no-instances with a fixed number of vertices or edges.

\bibliographystyle{alpha}
\bibliography{biblio.bib}

\end{document}